\documentclass[12pt]{article}
\usepackage[tmargin=1in,bmargin=1in,lmargin=1.25in,rmargin=1.25in]{geometry}

\usepackage{setspace}
\usepackage{csquotes}
\usepackage{kpfonts}
\usepackage{tikz-cd}
\usepackage{amsmath}
\usepackage{amsthm}
\usepackage{amsfonts}
\usepackage{stmaryrd}
\usepackage{dsfont}
\usepackage{graphicx}
\usepackage[title]{appendix}
\usepackage[font=small,labelfont=bf]{caption}
\usepackage{lipsum}
\usepackage{mwe}
\usepackage{comment}
\graphicspath{ {images/} }

\theoremstyle{definition}
\newtheorem{definition}{Definition}[section]

\theoremstyle{plain}
\newtheorem{theorem}{Theorem}[section]

\newtheorem{proposition}[theorem]{Proposition}

\title{Symmetries as Isomorphisms}
\date{}
\author{Lu Chen\thanks{chen.l@usc.edu}}
\begin{document}
	\maketitle

\begin{abstract}
	Symmetries and isomorphisms play similar conceptual roles when we consider how models represent physical situations, but they are formally distinct, as two models related by symmetries are not typically isomorphic. I offer a rigorous categorical strategy that formulate symmetries as isomorphisms between models and apply it to classical electromagnetism, and evaluate its philosophical significance in relation to the recent debate between `sophistication' and `reduction'. In addition to traditional spacetime models, I also consider algebraic models, in which case we can use the method of natural operators to address the problem of ontological nonperspicuity faced by the categorical strategy. Finally, I briefly expound on the significance of symmetries as isomorphisms in the framework of Univalent Foundations, in which isomorphic structures are formally identified.

\end{abstract}

\tableofcontents

\section{Introduction}

 Symmetries of a theory are transformations of its models that leave the relevant physics unaltered. They are typically not isomorphisms, which are structure-preserving maps between two structures (typically of the same kind).  In this paper, I will provide a strategy for reformulating symmetries between models as isomorphisms. There are several reasons for doing this. First, symmetries and isomorphisms play similar conceptual roles: symmetry-related models (SRMs) are considered to represent the same physical situation just like isomorphic models, and both are a guidance for identifying redundancies in our models or theories (see, for example, Dasgupta [2016]). It would therefore be attractive to unify their mathematical forms. Second, the nature of isomorphism is relatively well understood: it is as far as one can go in identifying a mathematical model (the role of which is highlighted by Wallace [2022]). Indeed, isomorphic models are formally identified in Univalent Foundations (UF), an alternative framework to standard set theory.\footnote{\label{Uf} The research presented in this paper was initially a section of a more extensive work on UF. I will briefly discuss UF in Section 5. While it helps to know the background of UF, I hope the paper nevertheless makes a stand-alone contribution.} Thirdly, reducing symmetries to isomorphisms is not the same as dispensing with them altogether (which would be the case if we reduce symmetries to identity maps). This is crucial for preserving the important role of gauge symmetries in physics (see, for example, Nguyen et al [2020], Tong [2018]). Moreover, the models resulting from systematically reformulating symmetries as isomorphisms are not theoretically equivalent to other versions, and therefore enrich our understanding of the relevant physics. Even when such a reformation is not necessary, it is beneficial to expand our understanding by exploring alternative modelling approaches with their distinct features.
 
There is a recent related debate between \emph{sophistication} and \emph{reduction} in the recent literature on symmetries (in Dewar's [2019a] terminology). According to sophisticationists, SRMs \textit{invariably} represent the same physical situation (see Martens and Read [2021]). Reductionists in contrast hold that, for any SRMs, we need to identify a reduced model that captures their common ontology in order to claim that these SRMs really do represent the same physical situation.  Dewar further distinguishes between `external sophistication' and `internal sophistication'. According to the former, SRMs are isomorphic models `by fiat', or in other words,  we can treat SRMs `as if' they were isomorphic models. Sophistication then follows from the principle of \textit{sophisticated substantivalism} that isomorphic models represent the same physical situation. According to the latter,  we should formulate SRMs as models whose internal structures can be shown to be isomorphic under the symmetries. External sophistication has been heavily critiqued. Indeed, when sophistication is under attack, the attack typically targets at the external approach (Marten and Read [2021], Jacob [2022]; see also March [forthcoming]). In contrast, the internal approach aligns with the more familiar fibre-bundle approach to physics that many advocate (see, for example, Healey [2007] and Jacobs [2023]). 

External and internal sophistication shouldn't be regarded as two distinct views about symmetries but rather two complementary strategies for implementing sophistication, at least the way Dewar presents them. The external route is criticized for its opacity, while the problem with the internal route is that it may not be available and thus does not guarantee that SRMs are invariably physically equivalent.  My goal is to develop and defend sophistication (1) making external sophistication mathematically rigorous through a  category-theoretic strategy, and (2) introducing a framework of algebraic models in which we can systematically `internalize' external sophistication. The resulting models, however, should be distinguished from those fibre-bundle models featured in Dewar's internal sophistication. While my position is broadly sophisticationist, I should note in advance that it differs in detail from the positions of sophisticationists presented in the literature, which will become clear later.\footnote{My main position is to defend SRMs invariably represent the same physical situation by defending the category-theoretical strategy of unifying the formal treatment of symmetries with familiar isomorphisms. My position is not that models for which symmetries are (nontrivial) isomorphisms are better models of reality. This may be true, but it should not be decided solely on the ground of the nature of symmetries.} 

The category-theoretic strategy in (1) is not completely novel (see Weatherall [2016]).\footnote{The formalism I appeal to is slightly more formal than what is presented by Weatherall.} But it is often presented in a way that conflates external sophistication with the internal one, namely the fibre-bundle picture, or with the quotient view that appeals to equivalence classes of SRMs.\footnote{In Weatherall [2016], the category of `external models' is shown to be equivalent to the category of fibre-bundle models (i.e.,`internal models'). This further conflates the two strategies under the criterion of theoretical equivalence based on categorical equivalence.  } Thus, apart from presenting a rigorous strategy for external sophistication, I also aim at peeling apart various approaches that are usually conflated but should be held as distinct (Section 3).


The `internalization' strategy in (2) is based on the algebraic framework developed by Chen and Fritz ([2021]). A distinctive aspect of algebraic models is that we can further employ a mathematical procedure to reveal their internal structures, which is not the case for traditional models. Traditional spacetime models take form in $\langle M, g, \{\phi\}\rangle$, where $M$ is a spacetime manifold, $g$ the metric field, $\{\phi\}$ other physical fields. 
After we reformulate symmetries as isomorphisms according to our strategy, we do not have a way to know the internal structure of the resulting models.\footnote{This focus on external behaviors of models under symmetries instead of their internal structures is reminiscent of the Klein tradition of geometry, according to which geometry is solely characterized by its symmetry group (Klein [1892], see also Norton [1993]). This tradition has been defended, for example, by Wallace ([2019]). The first part of the paper can be seen as engaging with this tradition, generalizing from geometry to spacetime models.} This is criticized as not ontologically perspicuous, among other criticisms.\footnote{I do not think, however, this is a decisive criticism against the category-theoretic strategy, as I will explain later.} 

Algebraic models can help overcome this problem (to the extent that they are independently feasible). These models take form in $\mathfrak{A}(\Psi,\Phi, ...)$, where $\Psi, \Phi,...$ are all possible configurations of physical fields and $\mathfrak{A}$ is the algebraic structure that exhaustively characterizes their structural relations with each other.\footnote{Here, I focus on the new literature on algebraic models starting with Chen and Fritz ([2021]). The old algebraic models introduced by Geroch ([1972]) are not substantially different from traditional spacetime models (see Rosenstock et al [2015]). Note that the capital Greek letters here represent physical fields that consist of \emph{all} their field configurations. In contrast, the lower-cased Greek letter used for the spacetime models refers to a particular field configuration.} Such models depart from traditional spacetime models in that they no longer posit an underlying spacetime manifold on which physical fields are defined. Regarding these models, there is a mathematically well-defined procedure of obtaining the algebraic structure invariant under given symmetries based on the tool of `natural operators' (Section 4; see Kolar et al [1993]).

In Section 5, I will further highlight the significance of isomorphisms---and thereby of sophistication---by appealing to UF, a new logico-mathematical framework in which we formally identify structures that are isomorphic. In UF then, symmetries are reconstrued as identity maps (where `identity' is a technical notion in UF, which I argue elsewhere to be interpreted realistically), which gives a rigorous implementation of the idea that symmetries in question are indeed representational redundancies. (Regarding technical background, this paper is suitable for readers with a general understanding of category theory and mathematical physics.\footnote{One of my future plans continuing from this project is to engage with higher categories, which can be a powerful tool for higher symmetries (see Bhardwaj et al [2023]). In particular, it would be useful for discussing higher gauge theory which expands our theoretical framework for gravity (see Baez and Wise [2014]). })

\section{Examples of Symmetries}

Let me start with the simple case of a global velocity boost in Newtonian mechanics. Let  $\langle E, v\rangle$, $\langle E, v'\rangle$ be models related by this velocity boost, where $E$ is a Euclidean space, and $v,v'$ are maps from point particles (a subset of the point set in $E$) to velocity vectors such that $v-v'= constant$. The two models are not isomorphic since except for very specific $v$ and $v'$, there is no automorphism $f$ on $E$ that takes $v$ to $v'$ (namely $v'=f_*v$, where $f_*$ is the pushforward of $f$). According to the standard approach, such symmetry-related models represent the very same physical situations. The difference between the models reflects a discrepancy between our representation and what is  being represented, namely that there is a surplus structure in our models. It is thus desirable to dispense with the surplus structure. A typical approach is to come up with reduced models that do not have the variant quantities under symmetry transformations. We can achieve this, for example, by simply substituting the velocity field with an acceleration field.




Let's now turn to gauge theories, which will be the main focus of the paper.\footnote{This paper primarily discusses classical gauge symmetries. But the underlying principles should apply to other formal symmetries as well, such as the CPT symmetries in quantum theory. I intend to focus on these quantum cases in the future.} Classical electromagnetism is an exemplary example and indeed much discussed in the literature that I engage with.  Readers familiar with the literature should feel free to skip Section 2.1, and only refer back to it when necessary.\footnote{Note that my presentation of the background is more detailed than standard in this corner of the literature, because I aim at more accessibility for technically minded readers who don't have the background and wish to pick up the relevant physics along the way. }

\subsection{Four Models of Electromagnetism}

 The electromagnetic potential field, represented by  a four-covector $A_\mu$ (where $\mu$ is the index for temporal and spatial components), plays a crucial role in determining the behavior of electromagnetic fields. However, a key insight of classical electromagnetism is that the absolute value of this potential is not directly observable. This means that different values of  $A_\mu$  can lead to the same physical predictions as long as their differences can be accounted for by a gauge transformation.   Let $\phi$ be any scalar field. Let $x$ range over spacetime points. Then a gauge transformation on the field can be written as $A_\mu(x)\to A_\mu(x)+\partial _\mu \phi(x)$. (This is a `local symmetry' in the sense that we can transform the value at each spacetime point differently, as opposed to a `global symmetry' like velocity boost, where the value as each spacetime point is changed uniformly.) The Maxwell equations that govern classical electromagnetism are invariant under gauge transformations.  Note that, like in the case of velocity boost, models like $\langle M, A_\mu\rangle$ and $\langle M, A'_\mu\rangle$, where $M$ is a Lorentzian manifold, related by a gauge transformation are not isomorphic. 
 
To dispense with this representational redundancy of $A_\mu$ models, we can formulate a reduced model invariant under gauge transformation by appealing to the electromagnetic force field $F_{\mu\nu}=\partial_\mu A_\nu-\partial_\nu A_\mu$, which is a rank-2 antisymmetric tensor.   It is easy to check that $F_{\mu\nu}$ is indeed invariant under the gauge transformation $A'_\mu= A_\mu+\partial _\mu \phi$.\footnote{$F'_{\mu\nu} = \partial_{\mu}A'_{\nu} - \partial_{\nu}A'_{\mu}
= \partial_{\mu}A_{\nu} - \partial_{\nu}A_{\mu} + \partial_{\mu}\partial_{\nu}\phi - \partial_{\nu}\partial_{\mu}\phi
= F_{\mu\nu}.$} Therefore, the resulting model $\langle M, F_{\mu\nu}\rangle$ is invariant under gauge transformations.\footnote{There are, however, controversies on whether $F_{\mu\nu}$ exhausts the physical content of the electromagnetic field, since it may not encode all the physically relevant information invariant under the gauge transformations. See Healey [2007] and Dewar [2019a]. We will return to this in Section 3.} 

The mathematics of electromagnetism can be equivalently written in terms of differential forms, which is mathematically more elegant and will be useful for our discussions. In this framework, the potential field is represented by a differential 1-form $A$, with gauge transformations expressed as $A\to A+d\phi$, where $d$ is exterior derivative and $\phi$ a scalar field (also a 0-form). $F$ is a 2-form such that $F=dA$. Then, we have $F'=dA+d(d\phi)=dA=F$.\footnote{Nilpotency, namely $\mathrm{d}(\mathrm{d}\omega) = 0$ for any form $\omega$, is a defining condition for exterior derivatives.} So again, $F$ is invariant under gauge transformations. 

Note that gauge transformations are not (nontrivial) isomorphisms for either $A$-models or $F$-models. They are not isomorphisms in the case of $A$-models for the similar reason as in the case of velocity boost. As for $F$-models, a gauge transformation on $A$ is simply an identity map. (While diffeomorphisms on $M$ can induce isomorphic models to  $\langle M, F\rangle$, these are not relevant for the symmetries currently under discussion.) 

Like the case of velocity boost, there is another route to reformulate electromagnetic models. Indeed, we can reformulate gauge symmetries as isomorphisms by resorting to the framework of principal bundle (PB). In the PB picture, the potential field can be formulated as a (Lie-algebra-valued) connection 1-form $\omega$ on the $U(1)$-principal bundle $P$, which in the simplest case is the product space $S^1\times M$ (which consists of a circle for each spacetime point---thus also known as `a circle bundle'). A `connection' on the bundle indicates how one point in a circle connects to another point in a `neighboring' fiber and thus how to move in the bundle (that is, it takes a tangent vector as the input, which is like an infinitesimal path, and gives a Lie-algebra element as the output, namely how an element of a fiber transforms to another along this path; see Appendix A). In this picture, the gauge transformations can be visualized as arbitrary smooth rotations of these circle fibers (like twisting the Slinky). Since there is no fact as for which points in different fibers are really `neighbors', this picture clearly illustrates  gauge transformations as merely a feature of our representational redundancy. Note that the resulting $\omega$-models $\langle P, \omega\rangle$ and $\langle P,\omega'\rangle$ related by a gauge transformation are isomorphic (Appendix A). Unlike $F$-models, each nontrivial gauge transformation is a nontrivial isomorphism.

The two approaches can be combined. The electromagnetic field strength $F$ can be reconceptualized as a curvature 2-form $\Omega$ on the $U(1)$-principal bundle $P$ (formally $\Omega=d\omega +\omega\wedge \omega$). Intuitively, if we start at a point on a fiber, and move around in a loop according to the connection, we end up typically at a different point on the fiber from the starting point. When such a loop is `infinitesimal', we call the `angle' between the starting and the end point a `curvature'. A curvature is a 2-form because it takes two tangent vectors, which form an infinitesimal loop, and it determines how an element of a fiber transforms to another.  It is trivial to see that however we twist the circles (like our Slinky) along the loop, the relative positions of the starting and end points on the same circle do not change. So $\Omega$ is indeed invariant under gauge transformations.


Dewar ([2019a]) called the $\omega$-models `sophisticated' and the $F$-models `reduced'.  The reduced models are formulated solely in terms of notions that are invariant under the symmetries of the theory, while  \emph{sophisticated} models are variant but isomorphic under the symmetries. Accordingly, \textit{reduction} is the approach that we should find reduced models to capture the common ontology of symmetry-related models (SRMs) in order to justify the claim that SRMs represent the same physical situation, while \textit{sophistication} is roughly the approach that SRMs \textit{invariably} represent the same physical situation just like isomorphic models. Dewar further suggests \textit{external sophistication}, drawing on Wallace ([2019]), according to which we can declare SRMs as isomorphic `by fiat', or, as others put it, treat them `as if' they were isomorphic (Martens and Read [2021]). In contrast, what Dewar calls \textit{internal sophistication} refers to the alternative strategy of reformulating SRMs into new structures such that they can be shown to be isomorphic. The $\omega$ models are an example of it.

Reduction is often preferred over sophistication in the philosophical literature on symmetries, but Dewar argues in favor of the latter. To argue for this, Dewar points to the conceptual advantages of sophisticated models. For example, certain primitive facts about $F$, such as the axiom of anti-symmetricity  (i.e., $F_{\mu\nu}=-F_{\nu\mu}$) can be explained by  a trivial arithmetic statement about the potential $A_\mu$, namely $\partial_\mu A_\nu-\partial_\nu A_\mu = -(\partial_\nu A_\mu-\partial_\mu A_\nu)$.\footnote{This reasoning also works if $A_\mu$ is replaced by a connection on the principal bundle: we can define $A$ as a potential of the target connection $\omega$ relative to an arbitrary chosen flat connection as a reference. } 

Another reason Dewar mentions is that we might not always be able to find an adequate reduced theory---a set of invariant quantities that completely characterize the relevant physics.  In the case of electromagnetism, there are $\omega$-models that do not correspond to any $F$-models, which I will come back to in Section 3. Thus, if $\omega$-models correctly characterize the relevant physics, $F$-models would fall short of that. This problem, however, also applies to internal sophistication. 

I shall add that the alleged advantage of the reduced approach that it dispenses with the surplus in our representations is not as significant as it sounds.   When gauge symmetries are conceptualized as isomorphisms, there is no principled difference between sophisticated models and reduced ones regarding their surplus structures.  In particular, consider the two models $\langle M,F\rangle$ and $\langle M,F'\rangle$ which are related by a diffeomorphism and are isomorphic, as well as the two models $\langle P, \omega \rangle$ and $\langle P,\omega'\rangle$  related by a gauge transformation. These situations are perfectly analogous to each other regarding surplus structures. In both cases, the same physical situation is represented by a mathematical model up to isomorphism. 

\section{A Categorical Strategy}

We have seen examples for  reformulating symmetry-related models into isomorphic ones. But this has only been done on a case-by-case basis. The systematic strategy of declaring SRMs as isomorphic by fiat, on the other hand, is widely criticized as too opaque and even downright contradictory (see Martens and Read [2021]). In this section, I will present a rigorous category-theoretical strategy for systematically reformulating SRMs as isomorphic.\footnote{Readers who are acquainted with the discussions that appeal to category-theoretical tools (e.g., Weatherall [2017]) will find this proposal familiar. To these readers, I would like to point out what I find to be problematic in the existent presentations of this strategy and what I wish to contribute. First, the technical or formal distinction between different categories of models resulting from reformulating SRMs is usually glossed over. Second, this part of the literature is often tied up with the general attitude that categorical equivalence implies theoretical equivalence, which I wish to reject. Both of them contribute to a somewhat prevailing conflation of distinct approaches to SRMs that I wish to peel apart. I will come back to this later in the section.}

To illustrate Dewar's external sophistication further, we know that for two vector spaces, which are sets equipped with addition and scalar multiplication, a bijective linear transformation is an isomorphism. Suppose we are given algebraic structure $\mathbb{R}^2$ (which has a richer structure than mere vector spaces), and the set of all bijective linear transformations on $\mathbb{R}^2$. Then we can reduce $\mathbb{R}^2$  to a mere vector space by simply stipulating those maps to be isomorphisms (see Wallace [2019], Dewar [2019a]). In other words, we shear away all the structures that such maps do not preserve. 

This, however, is not mathematically rigorous. In this example, the resulting structure is a coarser-grained algebraic structure than what we start with. But in general, there is no guarantee that the mathematical objects resulting from such a by-fiat declaration are algebraic structures or even well-defined set-theoretic structures.\footnote{\label{Alg}By an `algebraic structure' I mean an unstructured set equipped with algebraic operators. Clearly, it is not generally the case that the invariant structure under a set of symmetries is an algebraic structure. For example, a spacetime model in the form of $\langle	M, g\rangle$ is not an algebraic structure, since $g$ is a field on spacetime points, not an algebraic operator mapping between elements of $M$ (and the same applies to the topological and differential structures of $M$). In contrast, groups, rings, algebras, lie algebras are all examples of algebraic structures. For an example that the resulting objects are not set-theoretic structures, the `infinitesimal space' $\Delta$ in the dual category of $C^\infty$-rings is not characterizable by standard set theory (see Moerdijk and Reyes [1981], Bell [2008]).
	
	Of course, the authors in question may have particular structures in mind without aiming at a full generalization. In particular, Wallace ([2019]), which Dewar ([2019a]) draws upon, argues that we can define $\mathcal{G}$-structured space in terms of a collection of bijections to $\mathbb{R}^n$ that can be transformed to each other through group $\mathcal{G}$. The point of this formalism is to argue for a `coordinate-dependent' notion of space that specifies how it behaves under coordinate transformations. But here we are aiming at a generalization of a different kind for a different purpose---$\mathcal{G}$-structured space won't do. 
} In what follows, I shall substantiate this strategy in a more rigorous manner using category-theoretic notions. 

The standard spacetime models in the form of $\langle M, \{\phi\}\rangle $ (where $\{\phi\}$ includes particular configurations of physical fields on manifold $M$, including a metric field).
We start with the category of all such models with diffeomorphisms as morphisms. Then, we add symmetries into the morphisms of the category such that the models related by symmetries are isomorphic and each (nontrivial) symmetry transformation determines a (nontrivial) isomorphism.  For concreteness, I will explain this strategy in detail in the case of classical electromagnetism, which can be generalized to other classical gauge theories.\footnote{\label{gen}
	As I will explain later, the strategy generalizes to other classical gauge theories by replacing the $U(1)$ group used in classical electromagnetism by other symmetry groups. I shall note that this strategy might not be practically called for, because other gauge theories such as Yang-Mills theories are already standardly formulated according to the principal-bundle formalism, in which the symmetries are already isomorphisms. Conceptually, however, it is instructive to note the difference between this strategy and the principal bundle approach, as I will explain later.}
	
	It is almost magical that we can turn non-isomorphic SRMs by simply enriching the category with the symmetries, as we will see in more detail soon. While technically plain, this strikes me as a remarkable trick. We are now licensed to say that, for example,  any gauge-related models are isomorphic in a technically precise sense. To demystify it, however, note that the majority of this trick is simply done by the basic ideas of category theory, namely the role of morphisms, and more specifically how we use them to define isomorphisms. Category theory basically provides a notion of isomorphisms (i.e., invertible morphisms) that unify the familiar notion of isomorphisms and the familiar notion of symmetries. So it is really no surprise that SRMs can be considered isomorphic in this framework; and it is the practical success and utility of category theory that lends justification and plausibility to this  magical trick.

\subsection{Revisiting Electromagnetism}

\begin{proposition}
 Symmetry-related models are isomorphic in the category of models that includes the symmetry transformations as morphisms.
\end{proposition}

Start with the category \textbf{Pot$_{EM}$} of all $A$-models for classical electromagnetism that consists of models $\langle M, A\rangle$, where $M$ is a smooth manifold, and $A$ is a 1-form on $M$ that represents the electromagnetic potential. Let $A'$ be a 1-form that relates to $A$ by a gauge transformation from the symmetry group $U(1)$. If the category of such models has only diffeomorphisms (or smooth maps, if we want to be more general) between the manifolds as morphisms, as typically is the case, then such models are not isomorphic. To render them isomorphic, we can add new morphisms constructed out of the gauge transformations to the original category. Let each morphism be a pair consisting of a diffeomorphism between the manifolds and a gauge transformation between the potentials. More formally, call the resulting category \textbf{Gau}$_{EM}$, which consists of:

\begin{itemize}
 \item Objects written in the form of $\langle M, A\rangle$;
 
 \item Any morphism $f$ from  $\langle M, A\rangle$ to  $\langle M', A'\rangle$ is written in the form of $\langle \phi:M\to M', s: M\to U(1)\rangle$, where $\phi$ is a diffeomorphism and $s$ takes value in the symmetries that take $A$ to $A'$, namely $A'=A+ds$, such that:\footnote{\label{ss} A further technical point: It is natural to require that if the symmetries $s,s':M\to U(1)$ both take $A$ to $A'$, then they should be identified. Given that $A'=A+ds=A+ds'$, we have $ds-ds'=d(s-s')=0$, which entails that $s-s'=constant\in U(1)$ if the manifold is connected. To identify all such symmetries, we can require the symmetry $s$ that constitutes a morphism to be $s:M\to U(1)$ \emph{modulo $U(1)$}. This way, for any $A,A'$, the uniqueness of their mutual transformations is ensured. }
 
 \begin{itemize}
 	 \item The identity morphism $Id_{\langle M, A\rangle}$ is $\langle Id_M, M\to Id_{U(1)}\rangle$. For brevity, I will just write it as $Id$.
 	 
 	\item Given  $f=\langle \phi:M\to M', s: M\to U(1)\rangle$ and $g=\langle \psi:M'\to N, t: M'\to U(1)\rangle$, we have the composite morphism $g\circ f= \langle \psi\circ\phi:M\to N, (t\circ\phi)s: M\to U(1)\rangle$.\footnote{Note that this strategy is not completely new to philosophers and indeed already presented by Weatherall [2017], where the relevant category is defined very similarly. The difference here is that instead of using a real-valued function $\phi$ (see Section 2.1) in defining the category as in Weatherall [2017], I appeal directly to the symmetry group $U(1)$ of electromagnetism. The purpose of this is for easy generalization to other gauge theories  (and the strategy does not obviously depend on $U(1)$ being abelian). See also Footnote \ref{gen}.}
 \end{itemize}

\end{itemize}

It is clear that in this category, the objects $\langle M, A\rangle$ and $\langle M', A'\rangle$ are isomorphic. To show this, we can show that $f= \langle \phi:M\to M', s: M\to U(1)\rangle$ has an inverse $f^{-1}$, meaning $f\circ f^{-1}=f^{-1}\circ f=Id$. Let $\phi^{-1}$ be the inverse of the diffeomorphism $\phi$. For any $x\in M'$, let $s^{-1}(x)=s(\phi^{-1}(x))^{-1}$. We then construct $f^{-1}$ to be simply $\langle \phi^{-1},s^{-1}\rangle$.  It is straightforward to check that $f^{-1}$ is indeed the inverse of $f$.\footnote{As a useful technical observation, we note that this category is \textit{the category of elements} of $\mathcal{A}$ construed as a functor from the category of manifolds with the same morphisms to \textbf{Set} that maps a spacetime manifold to the set of all possible field configurations (which are connection 1-forms in this case).
	
		Let $\mathcal{F}$ be a functor from a category $\mathcal{C}$ to the category of sets, $\textbf{Set}$. The \emph{category of elements} of $\mathcal{F}$, denoted by $\int \mathcal{F}$, is a category that encodes the elements of the sets that the functor $\mathcal{F}$ maps objects of $\mathcal{C}$ to.	The objects of $\int \mathcal{F}$ are pairs $(c, x)$, where $c$ is an object of $\mathcal{C}$, and $x \in \mathcal{F}(c)$ is an element of the set that the functor $\mathcal{F}$ maps $c$ to.	A morphism in $\int \mathcal{F}$ between two objects $(c, x)$ and $(d, y)$ is a morphism $f: c \to d$ in $\mathcal{C}$ such that $\mathcal{F}(f)(x) = y$. In other words, $f$ is a morphism in $\mathcal{C}$ that is compatible with the action of the functor $\mathcal{F}$ on the elements $x$ and $y$.
		
		 The functoriality of the field functor guarantees that it is invariant under the given symmetries, and the resulting category of elements are isomorphic under the given symmetries. What we have shown is an instance of this more general observation. }

Let's turn to the comparison of  \textbf{Gau$_{EM}$} with other categories of electromagnetic models. First, consider the reduced $F$-models. Formally, let's define a category of models \textbf{Frc$_{EM}$} as consisting of $F$-models in the form of $\langle M, F\rangle$ as objects, and diffeomorphisms as morphisms. 

\begin{proposition}
\textbf{Gau$_{EM}$}  is equivalent to \textbf{Frc$_{EM}$} if spacetime has a trivial topology.\footnote{This is very close to Proposition 5.5 in Weatherall [2016].}
\end{proposition}
 To prove, we need to show that there is a functor from one to the other that is fully faithful and essentially surjective.\footnote{Briefly, a functor is (1) faithful if it preserves the distinctness of morphisms; (2) full if it maps to every morphism in the target category; (3) essentially surjective if every object in the target category is isomorphic to the image of some object in the source category. } Here's how to construct such a functor. Let a functor $\Psi: \textbf{Gau$_{EM}$}\to \textbf{Frc$_{EM}$}$ map each $\langle M, A\rangle$ to $\langle M, dA\rangle$ and map each morphism simply to its first element. $\Psi$ is indeed functorial since---first of all---the gauge transformation from $A$ to $A'=A+ds$ would leave $F$ invariant ($dA=dA'$ because $d$ is nilpotent), and moreover, a diffeomorphism pushes $F$ in the same way as it pushes $A$ (that is, $\phi^*F=d(\phi^*A)$, where $\phi^*$ is the pullback of a diffeomorphism  $\phi$ of interest; in particular, if $\phi^*A=A$ then $\phi^*F=F$). 

Suppose $f\neq g\in Hom(X,Y)$ for $X,Y\in \textbf{Gau$_{EM}$}$. Suppose for reductio that  $\Psi(f)=\Psi(g)$. If $f,g$ include different diffeomorphisms, then it immediately follows that $\Psi(f)\neq\Psi(g)$. Thus, they have the same diffeomorphism. But in this case, we have also ensured that the symmetry is unique if the manifold in question is connected (Footnote \ref{ss}), and therefore $f=g$, which contradicts the assumption.  Therefore, $f\neq g$ implies $\Psi(f)\neq\Psi(g)$, which means that $\Psi$ is faithful, provided that all manifolds in the categories are connected. It is easy to confirm that the functor is full and essentially surjective if spacetime is topologically trivial, since every curvature $F$ can be obtained from a connection $A$ by $F=dA$. Thus the categories are equivalent under the assumption that our spacetime is simply connected. 

It is worth noting that \textbf{Gau$_{EM}$} and \textbf{Frc$_{EM}$} are not isomorphic categories even if the topology is trivial, which is a stricter condition than categorical equivalence. Without getting into the formal details, we note that the functor does not preserve \emph{all} details of \textbf{Gau$_{EM}$}. For example, consider  $f= \langle \phi:M\to M, s: M\to U(1)\rangle$ that maps $\langle M, A\rangle$ to $\langle M, A'\rangle$ and $h= \langle \phi:M\to M, t: M\to U(1)\rangle$ that maps $\langle M, A\rangle$ to $\langle M, A''\rangle$ with $s\neq t$. Clearly, all the objects are mapped to $\langle M, F\rangle$ with $F=dA$, and all the morphisms are mapped to  $\phi$ despite $s\neq t$. Furthermore, let $\overline{\textbf{Pot}}_{EM}$ be the category of models in the form of $\langle M, [A]\rangle$ where $[A]$ is the equivalence class of potentials related by gauge transformation, with diffeomorphisms as morphisms. Similarly, we have that  \textbf{Gau$_{EM}$}  is equivalent but not isomorphic to $\overline{\textbf{Pot}}_{EM}$.

But what if our spacetime topology is not trivial? We can in fact prove the following elegant technical result:

\begin{proposition}
 The functor $\Psi$ from \textbf{Gau$_{EM}$} to \textbf{Frc$_{EM}$}  is faithful iff the 0-cohomology group is trivial, full iff the 1-cohomology group is trivial, and essentially surjective iff the 2-cohomology group is trivial.\footnote{For brevity, I simply refer to reduced re Rham cohomology groups by `cohomology group'. 	
		See Footnote \ref{derham}.} 
\end{proposition}

A differential $n$-form $\omega$ is \emph{closed} if  $d\omega=0$, and is \emph{exact} if there is a ($n$-1)-form $\alpha$ such that $d\alpha=\omega$. The $n$-th de Rham cohomology group can be expressed as the quotient of the set of all closed $n$-forms modulo the set of all exact $n$-forms.\footnote{\label{derham} More formally, for a smooth manifold $M$, the de Rham cohomology groups $H^k_{\mathrm{dR}}(M)$ are defined as the quotient of the closed $k$-forms by the exact $k$-forms:
	\[
	H^k_{\mathrm{dR}}(M) = \frac{\{\omega \in \Omega^k(M) \mid \mathrm{d}\omega = 0\}}{\{\eta \in \Omega^k(M) \mid \eta = \mathrm{d}\alpha \text{ for some } \alpha \in \Omega^{k-1}(M)\}},
	\]
	where $\Omega^k(M)$ denotes the space of smooth $k$-forms on $M$, and $\mathrm{d}$ is the exterior derivative. The reduced de Rham cohomology groups $\tilde{H}^k_{\mathrm{dR}}(M)$ are defined as follows:
	\[
	\tilde{H}^k_{\mathrm{dR}}(M) =
	\begin{cases}
		H^k_{\mathrm{dR}}(M) & \text{if } k > 0, \\
		H^0_{\mathrm{dR}}(M) / \mathbb{R} & \text{if } k = 0,
	\end{cases}
	\]} In other words, cohomology groups are measures of how much closed forms fail to be exact. In the case of 0-cohomology group, the set of all exact 0-forms (namely scalar functions) is trivial, since 0-forms are of the lowest degree. Thus, the group is equal to the set of all closed 0-forms, which is equal to scalar functions with constant values on connected submanifolds. That is, it measures how many disconnected parts our spacetime has. As mentioned earlier, if our spacetime is connected, that the cohomology group is trivial and the functor $\Psi$ from \textbf{Gau$_{EM}$} to \textbf{Frc$_{EM}$} is faithful.\footnote{Note that the unreduced 0-cohomology group is $\mathbb{R}$ rather than $\{0\}$. See Footnote \ref{derham}.} 

A nontrivial 1-cohomology group implies the existence of closed, non-exact 1-forms on the manifold that do not vanish when integrated along a closed loop. In other words, there are non-contractible 1-dimensional loops in the spacetime. For example, if spacetime were a 2-dimensional torus, then the cohomology group would be generated by two non-contractible loops (one along the `short' circle and one along the `long' circle in the intuitive picture of a donut).\footnote{The 2-dimensional torus $T^2$ is the Cartesian product of two circles: $T^2 = S^1 \times S^1$.		The de Rham 1-cohomology groups of a 2d torus $T^2$ 		$H^1_{\mathrm{dR}}(T^2) = \mathbb{R}^2$.
	The calculation  is based on the Künneth formula, which computes the cohomology of the product of two spaces from the cohomology of the individual spaces.} If 1-cohomology group is trivial---if spacetime has no penetrating holes---then the functor from \textbf{Gau$_{EM}$} to \textbf{Frc$_{EM}$}  is full because for every diffeomorphism $\phi$ that takes $F=dA$ to $F'=dA'$, there is an $s$ with $A'=\phi^*A+d\phi^* s$. To illustrate: in the simplest case that $\phi=id$, we have $F=F'$, which implies that $A-A'$ is exact given the trivial 1-cohomology group, hence the existence of a desired gauge transformation. On the other hand, this implication does not go through when the group is nontrivial like in the case of the torus. One can check that $A-A'$ is typically not exact with $A,A'$ being distinct, nontrivial derivatives of the torus coordinates. 

Similarly a nontrivial 2-cohomology group implies the existence of closed, non-exact 2-forms on the manifold that do not vanish when integrated over a 2d surface. This means that there are non-contractible 2-spheres in the spacetime. For example, $\mathbb{R}^3-\{0\}$ would be such a (coordinate) space. In this case, $\Psi$ is not essentially surjective because there are curvature-forms $F$ that are not $dA$ globally for any $A$ albeit locally so. On the other hand, if the cohomology group is trivial, since every $F$ is a closed form (because F satisfies the Bianchi identity $dF=0$), every $F$ is also an exact form. Thus every object in \textbf{Frc$_{EM}$} is reachable by $\Psi$, which means that $\Psi$ is essentially surjective. 

In summary, if our spacetime is topologically trivial, then the two categories are equivalent since there is a fully faithful and essentially surjective functor between the two. If it is not, then the two categories are not equivalent. Should we be worried in the latter case? Is it a problem for our method for constructing \textbf{Gau$_{EM}$}  if the two categories are not equivalent? No. On the contrary, this inequivalence draws out the conceptual substantiveness of this method as distinct from simply appealing to the invariance under the symmetries in question.

But for all we know, our spacetime \emph{is} topologically trivial. So what exactly is the significance of this possibly counterfactual inequivalence for conceptualizing about our world? The answer is that the dynamical laws of our world, for all we know, exhibit nontrivial behaviors of gauge potentials like $A$ which can be \textit{effectively} described by nontrivial cohomology groups. In other words, the gauge degree of freedom enters into dynamical laws and has empirical effects as if a non-trivial topology (or cohomology groups) were present. For example, the Aharonov-Bohm effect, in which the gauge degree of freedom affects the phase of a particle, can be effectively captured by positing a hole in space at the region where $F$ is nonzero (but zero elsewhere).\footnote{	To explain: consider a double-slit experiment with electrons, in which a solenoid producing a magnetic field is placed between the two slits. The magnetic field is confined to the interior of the solenoid, so the electrons only experience the potential $A$ outside the solenoid. 			The \emph{Aharonov-Bohm effect} predicts that the interference pattern of the electrons will be shifted due to the presence of the vector potential, even though the magnetic field is zero outside the solenoid. This shift can be quantified by the phase difference accumulated by the electrons traveling through the two paths (see Appendix for the definition of the integration):
	
	\begin{equation}
		\Delta \phi = \frac{e}{\hbar} \oint_C A,
	\end{equation}
	where $e$ is the electron charge, $\hbar$ is the reduced Planck constant, and $C$ is a closed contour around the solenoid. Here, $A$ cannot be gauged away because of a nonzero curvature $F$ at the solenoid as if there were a hole there.
	
	Of course, I do not claim that a nontrivial cohomology group captures the physical reality underlying the Aharonov-Bohm effect, i.e., that there really is a hole, nor that the cohomology interpretation is equivalent to other interpretations. For example, we could posit non-local interactions between the external particle in question and the curvature field $F$ to account for the effect (see Healey [2007]). Nevertheless, cohomology group is an elegant mathematical tool that allows us to discuss the relevant formal features of gauge fields without involving complex dynamical laws.}

Here, we can see the feature of our method of reformulating symmetry-related models as isomorphic by adding symmetry-based morphisms into the category. On the one hand, this allows us to get rid of the gauge degree of freedom in the sense that in the category \textbf{Gau$_{EM}$}, we wouldn't be able to distinguish between isomorphic models $\langle M, A\rangle$ and $\langle M', A'\rangle$ related by a diffeomorphism and a gauge transformation.\footnote{We identify objects up to isomorphism in category theory. See Section 5 on Univalent Foundations for related discussion.}  On the other hand, the category is inequivalent to the one based on standard invariant quantities under the symmetries and can account for the empirical consequences of the gauge fields like Aharonov-Bohm effect. Thus we do not lose the theoretical merits of gauge theories (see Nguyen et al [2020]).\footnote{\label{Tong}  In addition, as Tong [2018] points out, presenting a theory using gauge fields like $A_\mu$ is often much more concise than using gauge-invariant quantities such as Wilson lines. Also, appealing to the gauge fields allows us to preserve other theoretical virtues like Lorentz invariance, locality and unitarity. Tong also points out that, at the same time, we must identify gauge fields related by gauge transformations. They are dynamically equivalent, and also necessary for removing pathologies in quantum theory (such as the negative norms in the Hilbert space). } 

It remains to be seen how \textbf{Gau$_{EM}$} compares with the $PB$-models, especially the connection $\omega$-models. Let \textbf{Con$_{EM}$} be the category of $PB$-models in the form of $\langle P, \omega\rangle $, where $P$ is a principal bundle over a smooth manifold and $\omega$ is a bundle connection, with principal bundle isomorphisms as morphisms.  I am content with making the following claim, although we can make more informative claims like the case of $M$-models:

\begin{proposition}
	
\textbf{Gau$_{EM}$} is equivalent with \textbf{Con$_{EM}$} if and only if all principal bundles in \textbf{Con$_{EM}$} have a trivial topology. 
\end{proposition}
By a trivial topology of a principal bundle, I mean that the principal bundle is homeomorphic to a simple product space of the base space $M$ and the symmetry group $U(1)$, where $M$ may or may not have a trivial topology. To sketch the proof, regarding the `if' part, a principal bundle is equipped with a projection map $\pi$ from the total space $P$ to the base space $M$ (see Appendix A). Given a 1-form $A$ on $M$, we can uniquely determine a 1-form $\omega$ on $P$ (with a trivial topology)  by letting $\omega (v)$ for all $v$ in a tangent space on $P$ equal $A(d\pi (v))$ where the function $d\pi$ amounts to extracting the `$M$-part' of $v$. Given a 1-form $\omega$ on $P$, we can determine a 1-form $A$ on $M$ by choosing a flat-connection path $p$ in $P$ (recall Section 2.1) and projecting down the values of the connection $\omega$ on the path ($A=\pi^*\omega_p$). This is possible only because for any $U(1)$ action on the path, we also have the corresponding symmetry morphism from $A$, and vice versa. That is, a diffeomorphism between base manifolds and a $U(1)$-valued map $s$ uniquely determine a bundle morphism, and vice versa. The functor being described is indeed fully faithful and essentially surjective.

 The equivalence breaks down if \textbf{Con$_{EM}$} admits $P$ with a nontrivial topology. The topology of $P$ depends on how the fibers are 	`glued together' over $M$. For example if the base manifold is a 2d torus, there are countably infinitely many isomorphic classes of $U(1)$-principal bundles on it.\footnote{The classification of principal U(1) bundles over a base manifold $M$ is determined by the second cohomology group with integer coefficients, $H^2(M, \mathbb{Z})$. Each element of $H^2(M, \mathbb{Z})$ corresponds to an isomorphism class of principal U(1) bundles over $M$. In particular, if $H^2(M, \mathbb{Z})$ is nontrivial, there exist topologically distinct U(1) bundles over $M$. For example, consider a two-dimensional torus $T^2$ as the base manifold. 	The cohomology group $H^2(T^2, \mathbb{Z})$ is isomorphic to $\mathbb{Z}$, which means there is one independent generator that corresponds to a nontrivial U(1) bundle over $T^2$. This generator is associated with the integer-valued winding number that characterizes the `gluing' of the U(1) fibers specified by the transition functions on the torus.} For connections on bundles in different classes, there are no U(1)-actions that relate them because a U(1)-transformation is homotopical to a trivial transformation. In particular, imagine that our Slinky is twisted one time around the circle, with two ends glued together. However much we twist this twisted Slinky, we cannot untwist it (without tearing it apart and re-gluing it). In this case, it is easy to see (at least informally) that the functor from \textbf{Gau$_{EM}$} to   \textbf{Con$_{EM}$}  cannot be full or essentially surjective, since the former does not include additional U(1)-classes on $P$ with non-trivial topologies (in other words, nontrivial and trivial holonomies are projected into the same connections on $M$).   

Since they are not equivalent in the case of nontrivial principal bundles, which category of models we should use for electromagnetism depends on whether the nontrivial topology is useful for the theory (which is indeed the case, at least when we generalize to Yang-Mills theories; see Healey [2007], \textsection 3.1). 

\subsection{Interpreting Categorical Equivalence}

Let us summarize the results so far. Besides the original category \textbf{Pot}, we have considered four categories \textbf{Con}$_{EM}$, \textbf{Frc}$_{EM}$, \textbf{Gau$_{EM}$} and  $\overline{\textbf{Pot}}_{EM}$.  Among them,   \textbf{Frc}$_{EM}$ and \textbf{Gau$_{EM}$} are  equivalent just in case spacetime is topologically trivial, and the same for \textbf{Con}$_{EM}$ and \textbf{Gau$_{EM}$}  just in case the principal bundle is topologically trivial. \textbf{Gau$_{EM}$} and  $\overline{\textbf{Pot}}_{EM}$ are  equivalent, though not isomorphic, while \textbf{Frc}$_{EM}$ and $\overline{\textbf{Pot}}_{EM}$ are isomorphic when spacetime is topologically trivial. (Of course, we can say many other true things about every pair of these categories, but this will suffice for now.)

Do two isomorphic categories of models represent equivalent theories? For example, assuming that spacetime is topologically trivial, do \textbf{Frc}$_{EM}$ and $\overline{\textbf{Pot}}_{EM}$ count as equivalent representations of reality? While a thorough discussion of this question require a separate treatise, I would like to briefly explain why I think the answer is no. According to the semantic or model-theoretic approach, a scientific theory is characterized by a class of models, or in our categorical framework, a category of models, one of which corresponds to reality according to the theory. Now, an \textit{interpreted} theory under the syntactic approach is one that offers the ontological and ideological interpretations of its fundamental principles and laws (e.g.,  the Maxwell equations in classical electromagnetism) such that we know what fundamental entities and relations the theory postulates.  Similarly, an interpreted theory under the model-theoretic approach is one represented by a class of \textit{interpreted} models, or in our framework, an interpreted category of models, which tell us \textit{how} they correspond to reality. This  at least has initial plausibility: many people prefer models in the form of $\langle M, F\rangle$ over $\langle M, [A]\rangle$ (all else being equal) because the latter `consists of' an equivalence class of a potential fields and is less ontologically perspicuous than the former (see, for example, Jacobs [2022]). By endorsing $\langle M, [A]\rangle$ literally and  realistically, we are committed to what the model consists of, including the existence of potential fields.  Similarly, when it comes to categories of models, it matters whether we interpret a category as consisting of $F$-models or $[A]$-models even though the the two kinds of interpretations are equally legitimate and behave the same when spacetime is simple---after all, an uninterpreted abstract category consists of only a bunch of arrows.  By endorsing a category of models realistically, we are commited to the entities and relations that we use in interpreting the category. (These interpretations are implicit in our discussion so far, since they are built into our definitions of the categories.) Thus, I do not consider categorical equivalence or even isomorphism as a sufficient condition for theoretical equivalence, in the fully interpreted sense (see Wallace [2022]; see also Lutz [2010], Halvorson [2012]).

While this is by no means a thorough discussion, we have at least something to work with in explaining the difference between $[A]$-models in $\overline{\textbf{Pot}}_{EM}$ and the $A$-models in \textbf{Gau$_{EM}$}. Some authors have interpreted sophistication as the quotient view, which involves $[A]$-models, and object to sophistication on this ground (see, again, Martens and Read [2021] and Jacobs [2022]). I will come back to this objection soon, but for now I shall note that  the quotient view is neither reduction nor sophistication as I construe it. It is not reduction because the $[A]$-models are not formulated solely in notions that are invariant under gauge transformations. It is not sophistication in the sense that the (nontrivial) symmetries are not construed as (nontrivial) isomorphisms between models---all SRMs are collapsed into a single model. Now, one may ask: given the categorical equivalence between  \textbf{Gau$_{EM}$} and  $\overline{\textbf{Pot}}_{EM}$, don't they represent the same theory? The above discussion suggests no. In defining  \textbf{Gau$_{EM}$}, we appeal to the notion of symmetries instead of that of equivalence classes. I think we should prefer \textbf{Gau$_{EM}$} to $\overline{\textbf{Pot}}_{EM}$ because the former is more physically natural in the sense that it only appeals to notions that are of great utility to physicists rather than (say) artificial or idle concepts (see, for example, Footnote \ref{Tong}).

While I do not take categorical equivalence as a sufficient condition for theoretical equivalence, it is a necessary one. The theories represented by \textbf{Con}$_{EM}$, \textbf{Frc}$_{EM}$, \textbf{Gau$_{EM}$}  are thus all different. \textbf{Frc}$_{EM}$ is less expressively rich than \textbf{Gau$_{EM}$}, while  \textbf{Con}$_{EM}$ is richer than \textbf{Gau$_{EM}$}. Just like we can object to \textbf{Frc}$_{EM}$ on the ground that it does not exhaust all the invariant structures of electromagnetism formulated via potential fields (see Dewar [2019a]), we \textit{may} object to  \textbf{Con}$_{EM}$ on the ground that it adds invariant structures to the original theory. To be clear, I am not objecting to the principal bundle approach to electromagnetism per se: on the contrary, due to its rich expressive power, it may be preferred on empirical grounds. The point is just that---apart from that this does not rely on a systematic method---the theory is physically different from the original one, while \textbf{Gau$_{EM}$} captures exactly the theory represented by original $A$-models minus  the gauge degree of freedom. 

\subsection{Category Theory and Ontological Perspicuity}

I mentioned before that the advantage of reduction of dispensing with the surplus structure is not as significant as it sounds (Section 2.1). Let's delve into this claim more. We may consider there being two  paradigms regarding modelling.  According to the `old' paradigm, there should be a one-to-one correspondence between models and physical situations so that a model contains exactly the amount of information of the physical situation it models.\footnote{The labelling of `old' and `new'  is not intended to reflect the choronological order of the views. }  According to the `new' paradigm, each physical situation is represented by a bunch of models that are considered physically equivalent, without collapsing these models into one or selecting a privileged representative.

The old paradigm is hard to implement in practice.  For example, many  hold some version of Leibnizian equivalence or diffeomorphism invariance for spacetime theories.  According to the principle of diffeomorphism invariance, every model that is diffeomorphically related to a solution to the relevant dynamical equations is also a solution. Since it is physically unmotivated to distinguish between these models, they are considered physically equivalent.  As such, our models have a minimal representational redundancy highlighted by the indispensable role of isomorphisms. 

For many then, the move from reduction to sophistication is a move within the new paradigm, featuring an expansion of the collection of models representing a physical situation. (This is, of course, all familiar from Dewar's [2019a] discussion. I have only made it more rigorous.)   Now, those who accept the new paradigm may still ask what justifies this expansion---after all, we can in principle keep expanding the class of models by adding more morphisms. The justification, as I mentioned earlier, involves the highly analogous and important roles played by familiar isomorphisms and symmetries, crystalized in the unified treatment of them in our category-theoretical framework. From the point of view of category-theoretic physicists, there is no formal or practical reason not to expand in this way.   In this sense the expansion results in a more `joint-carving' representation of reality---though the `joints' here belong to the class (or category) of all models, not the world.

But, make no mistake---I do not intend to defend this approach without reservation. I do agree with various authors that it faces the charge of ontological (and ideological) non-perspicuity like the quotient view, which other approachs like reduction and internal sophistication do not face or to a less degree.\footnote{Against the quotient view, Sider famously wrote: `[I]t is intuitively unsatisfying to give no answer to the question of why relations of equivalence hold. When multiple equally good ways to represent are available, it is natural to ask why that is, to ask what it is about reality that enables it to be multiply represented'. (Sider [2020], p.207) For those who share this view, as I mentioned before, both my strategy and [A]-models would be equally unsatisfying. 
	
	Sider also immediately noted a division of opinions: `the demand for an explanation of equivalence in terms of third languages as akin to the explanatory demands that motivate physics: in each case we seek to understand how phenomena arise from reality’s ultimate constituents. But foes of metaphysics will regard the demand for metaphysical understanding as being a perversion of the legitimate demand for scientific understanding.' (p.208) This division, as I see it, corresponds to the opposition between the old and new paradigms of modelling. 
	
	As various authors (such as Martens and Read [2021]) point out, this problem doesn't apply as much to internal sophistication like the principal bundle approach. In that case, we do know the structures preserved by  symmetries (see Appendix A). However, there is still the analogy: models have symmetry-variant parts in external sophistication, just like models have isomorphism-variant parts in internal sophistication or sophisticated substantivalism (see also Footnote \ref{ont_UF}, Section 5). }  This is the case despite the virtues of the approach such as its physical naturalness and its
 formal elegance within category theory. Since there is no systematic way to recover the internal (set-theoretic) structure of the SRMs reconstrued as isomorphic---we only know their external behaviors in the relevant category---this approach thus provides no obvious answer to the question what the structure of reality is according to the models. This is perhaps not a serious  defect for certain category theorists who vehemently uphold that an object is what it does. But to go into that route would require a radical departure from our usual demand of ontological perspicuity.\footnote{Things might be different if we adopt a completely different logical foundation as in Section 5.}  That is, we would have to embrace a fundamental characterization of reality through the structural relations \textit{between} models rather than within models, and this may be too exotic even for structuralists to swallow.

\section{Natural Operator Algebraicism}

When we apply the aforementioned strategy and appeal to objects that are covariant under symmetry transformations, we are in the dark for what the objects are in terms of their internal structure. This counts against the category-theoretical strategy despite its other appeals. However, if we appeal to algebraic models instead of traditional models (albeit this task being highly nontrivial), the internal structure of which can be made explicit algorithmically using the tool of natural operators. I shall call this approach `natural operator algebraicism.'\footnote{In Chen and Fritz ([2021]), the resulting algebraic objects are called `field algebras' and the approach `dynamical algebraicism'. I use different labels for different emphases.} Due to the limited space and technical complexity of this approach, I will only provide a \emph{very preliminary} treatment of this subject.\footnote{I believe the current bite-sized treatment of natural operator algebraicism to be appropriate for the aim and scope of the paper, and hope the details included are suggestive enough towards the intended philosophical moral. Readers are encouraged to consult other work such as Kolar et al [1993] and Chen and Fritz [2021], as well as develop the strategy I sketch here. } The upshot is to suggest that for algebraic models, we have a systematic way of discerning the algebraic structures invariant under the symmetries, and thus overcome the main problem of external sophistication. 

First, what are algebraic models? To illustrate, let's consider a simple world with only a (real-valued) scalar field $\Psi$. It has field configurations $\psi_1,\psi_2,...$, one of which is actualized---call it $\psi$. \footnote{In previous sections (except for Section 1), I didn't distinguish between field configurations and fields because the distinction is unnecessary (only field configurations were relevant).}  In the traditional framework, a physical field is a function on the spacetime manifold, and thus a typical model includes a manifold together with the field configuration. This in our case results in models of the form $\langle M,\psi\rangle$. Reading such a model literally and realistically, the fundamental ontology includes a substantival manifold and a scalar field. Geroch ([1972]) pointed out that instead of defining scalar field configurations as smooth functions on a manifold, we can characterize them entirely through their algebraic relations with each other and thereby dispense with the underlying manifold. These algebraic relations define a commutative ring often written as $C^\infty (M)$.  We can subsequently define other notions needed for physics up to general relativity without reference to manifolds (although in this hypothetical simple world, there isn't much physics to do).\footnote{Again, note that we are focusing on classical theories. In quantum field theories, the nature of models may be very different. Since all configurations have a certain `amplitude', there may not be configurations that can be singled out as actual.
} The resulting model would then take the form of $\langle C^\infty(M), \psi\rangle $, where $\psi$ can be considered as a $0$-ary operator that picks out an element of $C^\infty(M)$. This model is an algebraic structure, known as a `pointy ring'.   Reading this model literally and realistically, the fundamental ontology includes all possible field configurations, but no underlying spacetime, and the fundamental ideology includes all algebraic operators as primitive relations between these configurations. Some, such as Butterfield ([1989]), have found the ontological commitment to these possible entities and their relations highly problematic, but there are ways to lessen this worry.\footnote{For example, if one is primitivist about laws of nature, namely that laws are brute fact irreducible to spatiotemporal patterns of the world, one may assign some ontological status to the possible field configurations that the laws refer to. Moreover, if we turn to quantum  theory, the path integral formulation explicitly involves summing over all possible field configurations, which all have salient physical significance (e.g., have amplitudes that contribute to the values of observables). It is also worth noting that it is typical for relationalists to appeal to merely possible material bodies or processes (for example, see Belot [1999]).} For example, a field configuration can be considered as a state of the field, which can be quantified over using higher-order logic, which is independently motivated (see for example Bacon [2022]; I won't go into this discussion further.) Now, $\langle C^\infty(M), \psi\rangle $ is an example of algebraic models.  In general, algebraic models are exactly those that are algebraic structures.\footnote{\label{alg_str} To remind the reader (see Footnote \ref{Alg}): by `algebraic structure', I mean a set of elements equipped with algebraic operators that map between elements.  Traditional spacetime models are counterexamples---fields are not algebraic operators that map between spacetime points. Also, note that according to this definition,  full models for general relativity discussed in Geroch ([1972]) and other authors (e.g., Rosenstock et al [2015]) are not algebraic models, even though they are built on an algebraic structure, namely $C^\infty(M)$. This is because the metric field and other fields they define are neither algebraic elements nor operators. In contrast, full models discussed in Chen and Fritz ([2021]) are algebraic models, where the metric field is included in the algebraic structure on a par with other fields. }  Chen and Fritz ([2021]) develop this approach in a way that an algebraic model does not have to include scalar fields, and can include many types of physical fields all at once (such as 1-form field configurations, spinor field configurations, tensor field configurations). The resulting model would take the form of $ \mathfrak{A}(\Psi,\Phi,\Xi,...)$, where $\mathfrak{A}$ is the relevant algebraic structure characterizing different fields $\Psi,\Phi,\Xi,...$, including $0$-ary operators that pick out realized field configurations as explained before.

 The relevance of algebraic models to our discussion is that we have the method of natural operators from category theory for finding out all the algebraic structures in them given all the structure preserving morphisms between the models. Let's illustrate this method using a simple example of groups---a group is a simplest kind of algebraic structures. We start with the category of groups \textbf{Grp}, with all groups for objects and all group homomorphisms for morphisms. Given these homomorphisms, how do we know that the objects are groups---that is, how do we discern the internal structure of a group from the behavior of these external maps from it? To proceed: we can define all sorts of functors on \textbf{Grp}---for instance, a functor to the category of sets \textbf{Set} or to \textbf{Grp} itself. Furthermore, a `natural transformation' is a map between functors that preserves their behaviors (more detailedly, it is a family of morphisms that transform the objects in the codomain of the functors that commute with the corresponding morphisms in the codomain of the functors; see Figure 1 (left)).  Now, a natural operator is just a natural transformation except that it can take $n$ arguments rather than just one.\footnote{This formalizes the intuitive idea that the behavior of a natural operator on an algebraic object is determined by the algebraic structure, since `not natural' typically means `not unique' (for example, the scalar multiplication is a natural operator on a vector space while the cross product is not).} Importantly, we can recover the original algebraic structure by obtaining all the natural operators. For the category \textbf{Grp}, for any group $G$, the familiar group operator $*:G\times G\to G$ is precisely a natural operator on \textbf{Grp}.  One can check that it is a natural transformation from the functor that maps any group $G$ to $|G| \times |G|$ to the functor that maps any group $G$ to its underlying set $|G|$ (Figure 1 (right)).\footnote{The reason that we are utilizing the underlying set of a group is that the natural operator is not necessarily a homomorphism between groups and therefore should be construed as a map between sets (i.e., a morphism in the \textbf{Set}). This is because the natural operator does not necessarily commute with itself if the group in question is not abelian. For example, the map $(a,b)\mapsto ab$ and $(a',b')\mapsto a'b'$ do not necessarily commute with the group multiplication, since $(aa',bb')\mapsto aa'bb'$ which is unequal to $aba'b'$ is the group is nonabelian. (Note we omit the notation for the group operator for brevity following the convention.)}  

\begin{figure}[h]
	\centering
	\begin{tikzcd}
		F(X) \arrow[r, "\alpha_X"] \arrow[d, "F(f)"'] & G(X) \arrow[d, "G(f)"] \\
		F(Y) \arrow[r, "\alpha_Y"'] & G(Y)
	\end{tikzcd}
	\hspace*{8mm}
	\begin{tikzcd}
		{|G|}\times {|G|} \arrow[r, "\eta_G"] \arrow[d, ""'] & {|G|} \arrow[d, ""] \\ {|H|}\times {|H|} \arrow[r, "\eta_H"'] & {|H|}
	\end{tikzcd}
	\caption{(left) the generic commutative diagram for the natural transformation $\alpha$, where $F,G$ are functors from category $C$ to $D$, with $X,Y$ being objects in $C$; (right) the simplified diagram for the natural transformation $\eta$ giving rise to the natural group operator.}
\end{figure}
Hence, there is no problem recovering the defining operator $*$ for a group as a natural operator for \textbf{Grp}. Although a group is only equipped with one primitive operator, there can be many more natural operators. For example, $x\mapsto x^3$ ($x\in G$) is a natural operator in $G$ but is not primitive according to our familiar definition of a group. But we do not need to worry about this redundancy. We can consider an algebraic object to be conceptually equivalent to the collection of all its natural operators, rather than a minimal set of primitives, which can be arbitrary.

In this example, we already know the algebraic objects in question are groups, since they are what we start with for illustration, but the same method can be applied to any category of interest.\footnote{For another simple illustration that is slightly more analogous with our case of interest, we can start with the category of rings with ring homomorphisms for morphisms.  Now like adding symmetries, we can add morphisms that do not preserve the multiplicative structure of rings into the category. The result, as can be expected, is the category of groups, which can be revealed by the method of natural operators. Notably, for pointy algebraic structures mentioned earlier, $0$-ary operators (which we use to pick out an element representing actualized state) can also be recovered as natural operators, as long as all morphisms preserve the constants (which is expected since homomorphisms commute with \textit{all} operators).

For the underlying technical result and proof, see Yuan's ([2013]) blog entry on natural operations, in which he proves the equivalence of the category of models defined by natural operators and the original category of algebraic models equipped with a forgetful functor to \textbf{Set}. 
}    To apply this method to physical models, suppose we start with a category of algebraic models with redundant degree of freedom according to a symmetry. As before, we can render SRMs as isomorphic by adding new morphisms based on this symmetry into the category. Then, we can find out the desired algebraic structure of the new models that are invariant under the symmetries by applying the above method of natural operators. It is of course not a trivial task to obtain a natural operator for a given category, but it is at least a mathematically well-defined procedure.\footnote{However, to give any meaningful example of this would get too technically involved, and in any event may be premature given the preliminary stage of natural operator algebraicism in the literature. The following discussion will largely gross over the technical matter and focus on philosophical aspects that are relevant to our subject matter.}

I should clarify that the procedure described above does not reflect how we actually formulate algebraic models invariant under symmetries. That is, we do not start with the category of certain algebraic models and add symmetries. This is because, first of all, when formulating any algebraic model in the category-theoretical framework, we typically have already taken symmetries into account, and have already appealed to natural operators that commute with the symmetries. Secondly, symmetries are usually formulated in the framework of traditional models, and I am not aware of any existent literature that discuss symmetries for algebraic models.  So the first step of adding symmetries to algebraic models is neither practically necessary nor demonstrably feasible. 
What happens instead is the following, according to the procedure used in Chen and Fritz ([2021]). We start with the category of relevant spacetime structures \textbf{Spacetime}, and consider physical fields on spacetime as functors $F:\textbf{Spacetime}\to \textbf{Set}$, from which we obtain natural operators. The morphisms in \textbf{Spacetime} can be stipulated to include not only smooth maps but also the desired symmetries. The functors on such a category automatically respect all these symmetries.\footnote{More realistically, different fields may have different symmetries. If a given symmetry does not apply to certain physical fields, we can incorporate a forgetful functor into the field functors.} Then we can obtain natural operators with respect to these functors.  These operators can define algebraic structures that are invariant under the said symmetries, just as the reconstructed groups are invariant under the group isomorphism in \textbf{Grp}. 

Now for our purposes, let's turn to a more physically realistic case, without worrying about the technical details (although the following discussions are not technically precise, I hope to illustrate how the strategy works philosohically). Consider classical electromagnetism again. Suppose we start with the potential one-form field $A$, which is conceptualized as  a functor from $\textbf{Spacetime}$ to $\textbf{Set}$ as before, and take the gauge symmetries into account. Then, the algebraic structures invariant under  the gauge symmetries that we can obtain through natural operators would amount to those that characterize not the original $A$-field, but the equivalence classes of the field configurations related by gauge symmetries. That is, the elements of the resulting algebraic models  (ignoring other physical fields) can be represented by $[A_1],[A_2],...$, where $A_i$s are 1-form field configurations and $[A_i]$s are equivalent classes of them related by gauge symmetries.\footnote{This is intuitive given that we have shown that the category of $[A]$-models is equivalent to the category \textbf{Gau$_{EM}$}, and the method of natural operators does not distinguish between equivalent base categories.} As we recall from previous discussion, if we assume spacetime is simply connected, then these algebraic elements can be equally represented by $F_1,F_2,...$ where $F_i$s are 2-form electromagnetic field configurations that are gauge invariant.  

In this example, we can see the core philosophical move  in response to the problem of external sophistication quite easily despite the technical complexity. In the previous section,  we see that the $[A]$-models as well as the $A$-models in \textbf{Gau$_{EM}$} face the charge of ontological (and ideological) nonperspicuity. In particular, we cannot obtain the set-theoretic structure of models in \textbf{Gau$_{EM}$} so that we are left with no obvious answer to the question what structures of reality are according to these models. But we do have a clear idea of how these models behave (indeed, assuming that the category is physically sufficient, we have a good enough handle on them for doing all the relevant physics).  Now, the strategy of algebraic models is that we postulate certain behaviors of fields as part of the fundamental structure of reality, and thereby `internalize' external sophistication and dodge the charge that we do not know the structure of reality.\footnote{\label{int_ext} The basic idea of the algebraic approach is that these `behaviors' are sufficient for reasoning with field equations and making empirical predictions; see Geroch [1972]. 
		
		By `internalize external sophistication', I do not mean that this approach is internal sophistication, just like I do not think the $[A]$-model approach belongs to sophistication. In particular, in the setting of Chen and Fritz [2021], a symmetry does not induce an isomorphism, but is just an identity map as in reduction. } In the case of $[A]$-models,  part of the charge of ontological nonperspicuity is that it makes reference to $A$-field that is variant under symmetries and appeal to the notion of equivalence classes, neither of which seems suitable to be among the fundamental furniture of reality. But if we postulate the behaviors of $[A]$ as fundamental, then they are just as perspicuous as the behaviors of the $F$ field. No equivalence class or symmetry-variant entity needs to be invoked. 

It is important to clarify that algebraic models themselves are not sophisticated models: in the formulation above, the gauge symmetries do not induce (nontrivial) isomorphisms between algebraic models. This, I think, is an important distinction between my proposal and other sophisticationsists' in the literature (see Footnote \ref{int_ext}). Indeed we can say that they are reduced given that they are formulated entirely through algebraic operators that are invariant under the symmetries. This is not too surprising, since the purpose is to fix the problem of ontological nonperspicuity, that is, to discern the structure of reality, which is invariant under symmetries. Note the original isomorphisms between spacetime models (induced by manifold diffeomorphisms) \textit{also} do not induce (nontrivial) isomorphisms between algebraic models. This is well known in the literature: for example, in Earman ([1989]), it is noted that various spacetime models related by the hole transformations are mapped to the same algebraic model.\footnote{Earman's point is that since spacetime models related by the hole transformations are mapped to the same algebraic model, and (in)determinism is only plausibly formulated in terms of spacetime models, `the proliferation of Leibniz algebras does not threaten determinism' (Earman [1989, p.193]).}  My main point is this: if we take algebraic models realistically and endorse their algebraic structures as exhausting the structure of reality, then we have a way of knowing the invariant structure of reality when we follow the category-theoretic method of reconceptualizing symmetries as isomorphisms.  In this sense, the worry about ontological nonperspicuity is lessened.\footnote{ Now, one may  ask whether the fact that algebraic models we arrive at are not sophisticated counts against my symmetry-as-isomorphism approach. That is, is this approach dialectically unstable because we use models for which symmetries are not isomorphisms to address the problem of nonperspicuity faced by conceptualizing symmetries as isomorphisms?  I don't think the strategy is unstable or inconsistent if we take the algebraic models to be fundamental as explained in the text, unless we want to argue for the sort of sophistication that requires the fundamental models to be sophisticated. This is not my position as I have noted earlier, although I certainly do not \textit{object to} it (which can be achieved in the case of algebraic models by letting the base category consist of principal-bundle models rather than spacetime models). In any case, my main point that symmetries between fields on spacetime can be treated on a par with symmetries of spacetime is compatible with algebraicism.} Insofar as the algebraic models capture the ontology that \textit{explains} the symmetries in the category of spacetime models, we may consider the category as more fundamental than the category of spacetime models that we start with.\footnote{Spacetime models, then, can be considered `effective' models. But one may wonder why  we should care about symmetries between nonfundamental models. In general, I think nonfundamental models (or theories) are also important subject of philosophical reflection since, after all, we do not have any plausible fundamental models given the mystery of quantum gravity (see, for example, Williams [2019]). }

Technically, the success of this strategy would rise or fall with the algebraic approach itself, and given the preliminariness of algebraic models, I don't think we can be completely confident about the technical feasibility of this strategy. But putting technicalities aside, one might still wonder how or whether the algebraic approach is motivated independently of the strategy here.  I can't hope to be comprehensive here given the scope of the paper, but I can briefly mention some potential benefits below.\footnote{These benefits are associated with the algebraic approach in general, and it waits to be seen whether they can be meaningfully combined with the strategies in this paper---hence `potential benefits'. I have argued elsewhere in more detail the advantages of the algebraic approach. }  In the algebraic approach, the definition and reasoning with derivatives and vectorial quantities are more elegant than in the standard approach.  For example, Moerdijk and Reyes ([1991], \textsection IV-V) noted that in the standard approach, the exterior derivative of an (n-1) form is interpreted as the circulation along $n$-cube, which involves integration over the boundary of the cube. In contrast, in the algebraic approach, no integration is involved, with the exterior derivative interpreted as the circulation along an infinitesimal cube, directly defined through simple algebraic relations. This also leads to a simpler metaphysics of  tangent space, which now is considered as simply a part of physical spacetime, instead of (say) an abstract space whose interaction with physical entities is mysterious and needs to be explained (see Chen [2022]). Moreover, the algebraic approach allows generalizations that are not possible for the standard manifold-theoretic approach, among which non-commutative geometry is a prominent example and potentially useful for quantum physics (see Connes and Marcolli [2008]; Huggett et al [2021]). This is intuitive: since we take relations between fields to be primitive without an underlying manifold, those relations can possibly violate the restrictions imposed qua the manifold in the standard approach.

\section{Symmetry as Identity}

Finally, I would like to expound on the significance of isomorphisms by appealing to Univalent Foundations (UF, or homotopy type theory (HoTT)). Consider this principle:

\begin{quote}
	\textsc{Identity.} All isomorphic structures are identical.
\end{quote}
Mathematical structuralists have standardly embraced \textsc{Identity} for mathematical entities. But \textsc{Identity}  can also apply to physical structures, motivated by the increasingly structuralist trend in modern physics (see Ladyman and Ross [2007]). As we know, distinct isomorphic models are prone to causing philosophical puzzles such as the  hole argument. A common response to the hole argument is that the diffeomorphically related models represent the same physical situation. If we embrace \textsc{Identity}, then we can enforce this response and systematically resolve the disparity between distinct isomorphic mathematical models and identical physical situations they represent. Combining the reformulation of symmetries as isomorphisms together with \textsc{Identity} further allows us to apply this idea to symmetry-related models so that we can formalize symmetries as representational redundancies on a par with isomorphisms. As I argue elsewhere, among the existent logico-mathematical foundations, UF alone implements  \textsc{Identity} through the univalence axiom and structure identity principle (UFP [2013]; Appendix B).\footnote{As I argue elsewhere in favor of UF-implemented ontic structuralism, it is tricky to describe a \textit{pure} structure without appealing to \textsc{Identity}. For example, even if we get rid of ordinary individuals in our models as proposed in Dewar [2019b], we still face the problem of individuals at higher orders. }

Note that this approach is different from reducing symmetries to identity maps understood in our standard logico-mathematical framework.\footnote{It is worth emphasizing the important role that isomorphisms play in general mathematics. Besides the category-theoretic application that a mathematical structure is defined as being invariant under its isomorphisms, one most important application of isomorphism is to translate proven statements from one mathematical object to those of an isomorphic one, for which a direct proof would be more difficult. For example, $\mathbb{Z}/10\mathbb{Z}$ is isomorphic to  $\mathbb{Z}/2\mathbb{Z} \times \mathbb{Z}/5\mathbb{Z}$. It is easier to find multiplicative inverses in the latter than in the former.}   UF does not trivialize the notion of isomorphism by equating it with the usual notion of identity but rather enriches and expands the latter into the former (see, for example, Ladyman and Presnell [2018]).

A category $A$ (or rather `precategory'; see Appendix B) is called \emph{univalent} iff:

\begin{quote}
	For any $a,b:A$, $a\simeq b$ implies $a=b$. (UFP [2013], \textsection 9.1;  Appendix B)
\end{quote}
where $\simeq$ denotes isomorphism, which means that, like in standard category theory, there is a morphism $g : Hom_A(b, a)$ such that $g \circ f = 1_a$ and $f \circ g = 1_b$. This allows us to formally identify isomorphic objects that we cannot distinguish between in standard category theory. In a univalent category, every isomorphism is promoted to a distinct `way of identification', a notion that is unique to UF/HoTT (Appendix B; Ladyman and Presnell [2015]).\footnote{Here's a very quick explanation of identity claims in HoTT, assuming basic familiarity with type theory. First of all, an identity claim $a=_Ab$ for $a,b:A$ is a type. Importantly, the type can have more than one element, each of which is intuitively a way of identification. A useful interpretation of this claim is to understand $a,b$ as points (or paths, in higher-order cases) in a topological space and the identification as a path that connect them (or homotopies in higher-order cases).}

Thus, in UF, it is demonstrably true that for any $\langle M, A\rangle, \langle M,A'\rangle: \textbf{Gau}_{EM}$ with $A'=A+ds$, we have $\langle M, A\rangle= \langle M,A'\rangle$.\footnote{More generally, for any $\langle M, A\rangle, \langle M',A'\rangle: \textbf{Gau}_{EM}$ with $A'=A+ds$ and $M$ diffeomorphic to $M'$, we have $\langle M, A\rangle= \langle M',A'\rangle$.} The gauge transformation determined by $s$ exactly constitutes a way of identification between the `two'.\footnote{Here,  `two' refers to two ways of representation from an external point of view. From the internal point of view within UF, there is only one object. Let $=$ be referential identity and $\equiv$ be representational equality. In UF, we can assert $\langle M, A\rangle= \langle M,A'\rangle$ but not $\langle M, A\rangle\not\equiv \langle M,A'\rangle$. The latter is not even grammatical.} Thus we preserve the representational richness of the gauge field $A$ which plays important roles in theorizing, e.g., when spacetime has effectively nontrivial cohomology groups (Section 3). Meanwhile, we do not have to introduce the gauge degree of freedom into our ontology. When counting the objects in $\textbf{Gau}_{EM}$ in UF, we count all the gauge-related fields as one and the same (based on our standard paraphrase of cardinality in terms of identity). 
So we are not committed to an ontology any more extravagant than the one with only  gauge-invariant fields.  To emphasize, we may understand the merit of this UF-implemented approach as \emph{preserving the full expressive power of gauge theories without introducing gauge redundancies into one's ontology}.\footnote{\label{ont_UF} Unfortunately, there is no space here to fully discuss the issue of ontological perspicuity in the UF foundation, which is a subject of a separate treatise. But I by no means want to avoid this question. Briefly, here's the way I think about it, continuing from the discussion in Section 3. Recall that the ontological perspicuity problem with the $A$-models in $\textbf{Gau}_{EM}$ is that we cannot read the symmetry-invariant structures from the models. Adopting UF as the foundational framework may be considered as relinquishing precisely this requirement that a faithful representation of the structure of reality should be free of symmetry-variant elements (when symmetries are construed as isomorphisms). In this foundational language, we only pin down the structure of reality up to symmetries. This is analogous to how haecceities/individualities are treated in UF: we can use them to represent reality, without committing to their reality. Thus, to the question of what exist according to classical electromagnetism represented by $\textbf{Gau}_{EM}$, we can answer `the $A$-field' without worrying about introducing its surplus degree of freedom realistically.} 

Now, how does this differ from the standard practice of using gauge fields while denying the gauge degree of freedom as physical (see, for example, Luc [2023])? The problem  with the standard practice, again, is that the disparity between our representational apparatus and what we intend to represent is prone to causing philosophical confusions. In standard logico-mathematical framework, we have that  $\langle M, A\rangle\neq \langle M,A'\rangle$ (or just $A\neq A'$), and yet we stipulate that they represent the same physical situation. But this disparity can make one think: if the gauge degrees of freedom are so theoretically valuable, why not think that they do correspond to physical reality even if they are empirically underdetermined? One is then torn between ontological perspicuity and other theoretical virtues (see Section 3; see also Footnote \ref{ont_UF}). 
In the words of Wittgenstein: `most of the propositions and questions of philosophers arise from our failure to understand the logic of our language.' (Tractatus 4.003) The aspiration of the UF-implemented approach of symmetry-as-identity is precisely to cure the confusion caused by our inappropriate language.



\newpage

\appendix

\section{Principal Bundle}

In this section, I provide definitions and a proof relevant to Section 2 related to the notion of principal bundle. For more background on Lie-algebra-valued differential forms on manifolds or on principal bundles relevant for our discussion, the readers may turn to Burke ([1985]) and Nakahara ([2003]).

	\begin{definition}[Principal bundle]
		A \emph{principal bundle} is a quadruple $\langle P, M, G, \pi\rangle$, where $P$ is a smooth manifold called the \emph{total space}, $M$ is a smooth manifold called the \emph{base space}, $G$ is a Lie group called the \emph{structure group}, and $\pi : P \rightarrow M$ is a smooth surjective map called the \emph{projection}. The following conditions must hold:
		\begin{enumerate}
			\item For each point $m \in M$, the preimage $\pi^{-1}(m)$ is a copy of $G$.
			\item There is a smooth right action of $G$ on $P$ such that the projection map $\pi$ is $G$-equivariant, i.e., $\pi(pg) = \pi(p)$ for all $p \in P$ and $g \in G$.
			\item The principal bundle must be locally trivial, i.e., there exists a neighborhood $U \subseteq M$ of any point $m \in M$ and a diffeomorphism $\phi : \pi^{-1}(U) \rightarrow U \times G$ such that $\pi = \mathrm{pr}_1 \circ \phi$, where $\mathrm{pr}_1 : U \times G \rightarrow U$ is the projection onto the first factor. Moreover, the diffeomorphism $\phi$ must be $G$-equivariant, i.e., $\phi(p g) = (\phi(p))g$ for all $p \in P$ and $g \in G$.
		\end{enumerate}
	\end{definition}
	
\begin{definition}[Connection]
	For principal bundle $\langle P, M, G, \pi\rangle$, a connection 1-form on $P$ is a $\mathfrak{g}$-valued 1-form $\omega \in \Omega^1(P, \mathfrak{g})$ satisfying the following two conditions:
	
	\begin{enumerate}
		\item $\omega$ is $G$-equivariant, i.e., for any $g \in G$ and $X \in T_pP$ (the tangent space at point $p$ in $P$), we have $(R_g^* \omega)_p = \mathrm{Ad}(g^{-1}) \omega_{p g}$, where $R_g$ denotes the right action of $G$ on $P$, and $\mathrm{Ad}$ is the adjoint representation on $\mathfrak{g}$.
		\item For any $\xi \in \mathfrak{g}$ (the Lie algebra of $G$) and $X_\xi$ the vector field on $P$ associated to $\xi$ by differentiating the $G$ action on $P$,  $\omega(X_\xi) = \xi$. (Basically, $X_\xi$ represents how each point in the bundle would infinitesimally move  according to $\xi$. This condition makes sure that the structure of $P$ captured by connection $\omega$  is compatible with the dictation of $\mathfrak{g}$.)
	\end{enumerate}

\end{definition}

	\begin{theorem}
		Given a principal bundle $\langle P, M, G, \pi\rangle$, any gauge transformation between two connections $\omega$ and $\omega'$ is an automorphism on $P$.
	\end{theorem}
	
	\begin{proof}
	A gauge transformation between the connections  $\omega$ and $\omega'$ is given by a smooth map $g : P \rightarrow G$ such that
		\[
		\omega' = g^{-1} \omega g + g^{-1} \mathrm{d}g,
		\]
		where $\mathrm{d}$ is the exterior derivative, such that $g$ is $G$-equivariant: 
				\[hg(ph)=g(p)h\]
		for all $h\in G$ and $p\in P$.
		
	Let $\tilde{g}$  be a map from $P$ to $P$ such that $\tilde{g}(p)=pg(p)^{-1}$. We can show that $\tilde{g}$ is an automorphism on $P$.	To show that we need to show that it preserves the bundle structure, i.e., it covers the identity map on $M$, commutes with the right $G$-action on $P$, and is invertible.
		
		1. $\tilde{g}$ covers the identity map on $M$: this is clear from the definition (a $G$-action only maps between points on the same fibers; see Definition 7(2)). 
		
		2. $\tilde{g}$ commutes with the right $G$-action on $P$. From 			\[hg(ph)=g(p)h,\] we have	\[hg(p)^{-1}=g(ph)^{-1}h.\] Thus, we have \[\tilde{g}(ph)=phg(ph)^{-1}=phh^{-1}g(p)^{-1}h=\tilde{g}(p)h\]
		
		3. $\tilde{g}$ is invertible. Let $\tilde{g}^{-1}(p)=pg(p)$. We have \[ (\tilde{g}\tilde{g}^{-1})(p)=\tilde{g}(\tilde{g}^{-1}(p))=pg(p)g(p)^{-1}=p\]
\end{proof}
		
		\section{Structure Identity Principle}
		
In this section, I sketch definitions relevant to Section 5 on Univalent Foundations (UF) based on the presentation in UFP ([2013]). Please see UFP for more background on the univalence axiom and the notions preceding that, which are less directly relevant for our discussion and too extensive to cover.

\begin{definition}[Identity]
		Given a type \( A \) and elements \( a, b : A \), the identity type, \( Id_A(a, b) \) or \( a =_A b \),  can be defined by the following rules:

\begin{itemize}
	\item \textbf{Formation Rule}: Given a type \(A\) and elements \(a, b : A\), there is a type \(Id_A(a, b)\).
	\item \textbf{Introduction Rule}: For any \(a : A\), there is a term \(refl_a : Id_A(a, a)\) representing reflexivity.
	\item \textbf{Elimination Rule}: For any family of types \(C(x, y, p)\) dependent on \(x, y : A\) and \(p : Id_A(x, y)\), given a term \(c : \prod_{a:A} C(a, a, refl_a)\), then  there is a term \(f: \prod_{x,y:A}\prod_{p:Id_A(x,y)} C(x, y, p)\) such that $f(x,x,refl_x)\equiv c(x)$.
\end{itemize}
Informally, \textbf{Introduction Rule} says that everything is self-identical, and \textbf{Elimination Rule} says that if $x=y$, then for any claim true of $x$, we also have a proof that it is true of $y$, namely the principle of indiscernibility of identicals.

\end{definition}

		To understand category theory in UF, we start with the notion of precategory, which corresponds to the notion of category in the standard approach. We will see soon why the term `precategory' instead of `category' is used.

		\begin{definition}[Precategory]
			A \emph{precategory} A consists of the following:
			\begin{itemize}
				\item A type \( \text{A}_0 \) whose elements are called \emph{objects}. We write $a:A$ for $a:A_0$.
				\item For each pair of objects \( a, b : \text{A} \), a type \( \text{Hom}_A(x, y) \) whose elements are called \emph{morphisms} from \( a \) to \( b \).
			\end{itemize}
			They satisfy the usual conditions, such as the composition rule for morphisms and the existence of identity morphisms. We write the identity morphism for $a:A$ as   $1_a:Hom_A(a,a)$.
		\end{definition}
		The notion of isomorphism between objects are defined similarly as in standard category theory:
		\begin{definition}[Isomorphism]
			A morphism $f : Hom_A(a, b)$ is an isomorphism if there is a morphism $g :
			Hom_A(b, a)$ such that $g \circ f = 1_a$ and $f \circ g = 1_b$. We write $f\simeq g$.
		\end{definition}

		Now, it is attractive to formally identify isomorphic objects in a category, since we can only identify objects up to isomorphisms in category theory. So, we say that  a category is a precategory that satisfies univalence:
		\begin{definition}[Category]
			A precategory $A$ is a category if for any $a,b:A$, $a\simeq b$ implies $a=b$. 
		\end{definition}

		Let's turn to structures. Let $X$ be a precategory. 
		
		\begin{definition}[(P,H)-structure]
			A notion of structure (P, H) over $X$ consists of the following:
			
			\begin{itemize} 
				\item A type family P : X$_0$ → U. For each x : X$_0$ the elements of Px are called (P, H)-structures on x.
				
				\item For x, y : X$_0$, f : Hom$_X$(x, y) and $\alpha$ : Px, $\beta$ : Py, a mere proposition
				H$_{\alpha\beta}$( f ).
				If H$_{\alpha\beta}$( f ) is true, we say that f is a (P, H)-homomorphism from $\alpha$ to $\beta$.
			\end{itemize}
			H satisfies further conditions that amount to the usual properties of homomorphisms.
		\end{definition}
		
		\begin{definition}
			A precategory of (P, H)-structures, $A = Str_{(P,H)}$(X), consists of the following:
			\begin{itemize}
				\item  $A_0:\equiv \Sigma (x:X_0) Px$.
				
				\item For $(x, \alpha), (y, \beta) : A_0$, we define	$Hom_A((x, \alpha), (y, \beta)) :\equiv	\{f : x \to y\mid	H_{\alpha\beta}( f )\}$
			\end{itemize}
			
		\end{definition}
		\noindent For example, $X$ can be the precategory of sets, and $A$ can be the precategory of groups, or rings, or topological spaces that are built on sets.

		Finally let me present the structure identity principle:
		\begin{proposition}[The structure identity principle]
			If $X$ is a category, then $Str_{(P,H)}(X)$ is a category.
		\end{proposition}
	\noindent It follows that isomorphic set-theoretic structures are identified in UF. This is because univalence entails that the precategory of sets is a category, in which all bijective sets are identified. For our purpose,  $Str_{(P,H)}(X)$ can be the category $\textbf{Gau}_{EM}$, whose objects are built on point sets. In this category, the symmetry-related $A_\mu$-models are identified.

\end{document}